\documentclass{article}

\usepackage[margin=1.25in]{geometry}

\usepackage{graphicx} 

\usepackage{mathtools}
\usepackage{amsmath,amssymb,amsthm}
\usepackage[labelfont=bf]{caption}

\usepackage{enumitem}

\usepackage[sc]{mathpazo}

\usepackage{bm}

\usepackage[dvipsnames]{xcolor}

\definecolor{c0}{HTML}{1d1d1d}
\definecolor{c1}{HTML}{173359}
\definecolor{c2}{HTML}{8f5a39}
\definecolor{c3}{HTML}{a8332b}
\definecolor{c4}{HTML}{4c825c}

\usepackage{algorithm,algorithmicx}
\usepackage[noend]{algpseudocode}

\algrenewcommand{\algorithmiccomment}[1]{\hfill \textcolor{gray}{$\vartriangleright$ \textit{#1}}}
\usepackage{etoolbox}
\algrenewcommand\alglinenumber[1]{\footnotesize #1:}
\makeatletter
\expandafter\patchcmd\csname\string\algorithmic\endcsname{\itemsep\z@}{\itemsep=5pt}{}{}
\makeatother

\usepackage{float}

\usepackage{booktabs}

\usepackage{hyperref}
\hypersetup{
    colorlinks,
    linkcolor=c2,
    citecolor=c2,
    urlcolor=c2
}
\usepackage[nameinlink,capitalize,noabbrev]{cleveref}
\crefformat{equation}{#2(#1)#3}
\crefmultiformat{equation}{{}#2(#1)#3}{ and~#2(#1)#3}{, #2(#1)#3}{ and~#2(#1)#3}

\usepackage{thm-restate}

\newtheorem{theorem}{Theorem}
\numberwithin{theorem}{section}

\newtheorem{lemma}[theorem]{Lemma}
\newtheorem{problem}[theorem]{Problem}

\newtheorem{remark}[theorem]{Remark}
\newtheorem{imptheorem}[theorem]{Imported Theorem}

\theoremstyle{definition}

\newtheorem{assumption}[theorem]{Assumption}

\numberwithin{equation}{section}

\let\originalleft\left
\let\originalright\right
\renewcommand{\left}{\mathopen{}\mathclose\bgroup\originalleft}
\renewcommand{\right}{\aftergroup\egroup\originalright}

\DeclareMathOperator{\range}{range}
\newcommand{\bv}[1]{\mathbf{#1}}
\newcommand{\T}{{\scalebox{.7}{$\mathsf{T}$}}}
\newcommand{\F}{{\scalebox{.7}{$\mathsf{F}$}}}
\newcommand{\E}{\mathbb{E}}
\renewcommand{\Pr}{\mathbb{P}}

\newcommand{\condF}{\kappa_{\F}}
\newcommand{\cond}{\kappa}
\newcommand{\const}{\mathrm{const}}

\newcount\Comments
\newcommand{\kibitz}[2]{\ifnum\Comments=1\textcolor{#1}{#2}\fi}

\usepackage[style=alphabetic,url=false]{biblatex}
\allowdisplaybreaks

\title{A simple analysis of a quantum-inspired algorithm for solving low-rank linear systems}
\author{
Tyler Chen\thanks{Equal contribution: \{\texttt{tyler.chen}, \texttt{lyle.kim}, \texttt{archan.ray}\}\texttt{@jpmchase.com}.}\and
Junhyung Lyle Kim\footnotemark[1]\and
Archan Ray\footnotemark[1]\and
Shouvanik Chakrabarti\and
Dylan Herman\and
Niraj Kumar\and
\\
Global Technology Applied Research, JPMorganChase, New York, NY 10001, USA
}
\date{}

\begin{document}
\sloppy

\maketitle

\begin{abstract}
We describe and analyze a simple algorithm for sampling from the solution $\bv{x}^* := \bv{A}^+\bv{b}$ to a linear system $\bv{A}\bv{x} = \bv{b}$.
We assume access to a sampler which allows us to draw indices proportional to the squared row/column-norms of $\bv{A}$.
Our algorithm produces a compressed representation of some vector $\bv{x}$ for which $\|\bv{x}^*  - \bv{x}\| < \varepsilon \|\bv{x}^* \|$ in $\widetilde{O}(\condF^4 \cond^2 / \varepsilon^2)$ time, where $\condF := \|\bv{A}\|_\F\|\bv{A}^{+}\|$ and $\cond := \|\bv{A}\|\|\bv{A}^{+}\|$.
The representation of $\bv{x}$ allows us to query entries of $\bv{x}$ in $\widetilde{O}(\condF^2)$ time and sample proportional to the square entries of $\bv{x}$ in $\widetilde{O}(\condF^4 \cond^6)$ time, assuming access to a sampler which allows us to draw indices proportional to the squared entries of any given row of $\bv{A}$.
Our analysis, which is elementary, non-asymptotic, and fully self-contained, simplifies and clarifies several past analyses from literature including [Gily\'en, Song, and Tang; 2022] and [Shao and Montanaro; 2022, 2023].
\end{abstract}

\section{Introduction}

In this paper, we consider algorithms for solving a (dense) linear system
\begin{equation}
    \label{eqn:linear_system}
    \bv{A}\bv{x} = \bv{b}, 
    \quad\text{where}\quad
    \bv{A}\in\mathbb{R}^{n\times d}
    ,\quad 
    \bv{b}\in\mathbb{R}^n
    ,\quad
    \bv{b}\in\range(\bv{A}).
\end{equation}
Specifically, we are interested in the minimum norm solution $\bv{x}^* = \bv{A}^+\bv{b}$, where $\bv{A}^+$ is the pseudoinverse of $\bv{A}$.

The most typical access model for \cref{eqn:linear_system} assumes that, given indices $r$ and $c$, we can observe the $(r,c)$-th entry of $\bv{A}$ and the $r$-th entry of $\bv{b}$ at unit cost, but that no other information about $\bv{A}$ is known a priori. 
In this case, classical algorithms from numerical linear algebra solve the system in $O(nd^2)$ time \cite{trefethen_bau_22}.
Randomized algorithms are able to solve \cref{eqn:linear_system}
in the optimal $O(nd + d^{\omega})$ time, where $\omega < 2.371\ldots$ is the matrix-multiplication exponent \cite{chenakkod_derezinski_dong_rudelson_24}. 
In the case an $\varepsilon$-approximate solution is acceptable, full-gradient methods like conjugate gradient (on the normal equations) or heavy ball momentum run in $O(nd \cond\log(1/\varepsilon))$ time, where $\cond := \|\bv{A}\|\|\bv{A}^{+}\|$ is the \emph{condition number} of $\bv{A}$.
Such methods make use of a full matrix-vector product each iteration, which leads to the $O(nd)$ dependence \cite{saad_03}.

If we are given fast access to certain key information about $\bv{A}$, then we might hope to design faster algorithms.
For instance, the randomized Kaczmarz algorithm \cite{strohmer_vershynin_08} assumes that we can sample indices proportional to the square row-norms of $\bv{A}$ at unit cost.
Given such access, randomized Kaczmarz runs in $O(d\condF^2\log(1/\varepsilon))$ time, where $\condF := \|\bv{A}\|_\F\|\bv{A}^{+}\|$ is the \emph{Demmel condition number} \cite{strohmer_vershynin_08}.
It always holds that $\condF^2 \leq \operatorname{rank}(\bv{A})\cond^2$, and hence randomized Kaczmarz may provide acceleration over full-gradient methods on problems for which $\condF$ is sufficiently small.

Motivated by the potential speedups offered by randomized Kaczmarz and related algorithms, one might wonder whether we can further reduce the dependence on the dimension $d$.
At first glance, this is not possible; even writing down an approximate solution to \cref{eqn:linear_system} requires $\Omega(d)$ time.
However,  we can hope to learn some information about the solution $\bv{x}$ in $o(d)$ time.
For instance, we might be able to learn a few particular entries of $\bv{x}$, or even the ``most important'' entries of $\bv{x}$.

\paragraph{Notation.}
The $i$-th standard basis vector is $\bv{e}_i$ and the identity is $\bv{I}$.
The spectral and Frobenius norms of a matrix $\bv{M}$ are $\|\bv{M}\|$  and $\|\bv{M}\|_\F$ respectively.
The Euclidian norm of a vector $\bv{x}$ is $\|\bv{x}\|$.
The value ``$\const$'' is an absolute constant that may change at every instance.
Given a distribution $\mathcal{D}$ and an integer $I$, $\mathcal{D}^I$ denotes distribution on length-$I$ tuples where each entry is drawn independently from $\mathcal{D}$. 
We use  $\widetilde{O}(\,\cdot\,)$ to hide poly-logarithmic dependencies in the input parameters.

\subsection{Access model, goal, and quantum inspired algorithms}

Concretely, the focus of this paper is on the following problem:
\begin{problem}\label{problem:sampling}
Find (a representation of) some $\bv{x}$, for which $\|\bv{A}^{+}\bv{b} - \bv{x}\| < \varepsilon \|\bv{A}^+\bv{b}\|$, that efficiently admits the following access:
\begin{itemize}[itemsep=.1em,topsep=-.5em]
    \item Given an index $\ell$, return $\bv{x}_\ell$ 
    \item Sample an index $\ell\in\{1, \ldots d\}$ according the distribution $\mathcal{D}_{\bv{x}}$:
    \begin{equation*}
    \ell\sim \mathcal{D}_\bv{x}
    \qquad\Longleftrightarrow\qquad 
    \Pr[\ell = j]
    = \frac{|\bv{e}_j^\T\bv{x}|^2}{\|\bv{x}\|^2}
    ,\quad j=1,\ldots, d.
    \end{equation*}
\end{itemize}
\end{problem}

\Cref{problem:sampling} originates from the study of \emph{quantum-inspired} algorithms, which can be viewed as a classical analog of quantum algorithms; see \cite{tang_23} for an overview.
In short, while many quantum algorithms attained exponential speedups over the (at the time) best-known classical algorithms, the type of solutions they produced were not entirely comparable to the solutions produced by classical algorithms.
For instance, in the context of linear systems, the famous Harrow-Hassidim-Lloyd (HHL) algorithm is often stated to be exponentially faster than classical algorithms for solving linear systems; the runtime of HHL scales with the dimension $d$ as $O(\log(d))$ while classical algorithms scale as $\Omega(d)$ \cite{harrow_hassidim_lloyd_09}.
However, the sense in which HHL `solves' a linear system is more analogous to \cref{problem:sampling} than producing an approximation to the solution vector $\bv{x}^*$.
Thus, it is unfair to compare HHL to algorithms which are designed to produce the entire solution vector.


As noted above, there is no hope of efficiently solving \cref{problem:sampling} unless we assume some sort of useful access to $\bv{A}$.
In the context of quantum-inspired algorithms, a typical access assumption is the following:

\begin{assumption}[Sample and query access]
\label{asm:sample_query}
We assume that we can \emph{query} key properties about $\bv{A}$ and $\bv{b}$ at unit cost:
given indices $r$ and $c$, we can observe $\bv{e}_r^\T \bv{A} \bv{e}_{c}$, $\bv{b}_r$, $\|\bv{e}_r^\T\bv{A}\|^2$, and $\|\bv{A}\bv{e}_c\|^2$.
We also assume that we can \emph{sample} according to certain importance distributions relating to $\bv{A}$ at unit cost:
\begin{itemize}[itemsep=.1em,topsep=-.5em]
\item We can sample a row of $\bv{A}$ proportional to the squared row norm distribution $\mathcal{D}^{\textup{row}}$:
\begin{equation*}
r\sim  \mathcal{D}^{\textup{row}}
\qquad\Longleftrightarrow\qquad 
p_r := \Pr[ r = i ]
= 
\frac{\| \bv{e}_i^\T\bv{A} \|^2}{\|\bv{A}\|_\F^2} 
,\quad i=1,\ldots, n.
\end{equation*}
\item We can sample a column of $\bv{A}$ proportional to the squared column norm distribution $\mathcal{D}^{\textup{col}}$:
\begin{equation*}
c\sim  \mathcal{D}^{\textup{col}}
\qquad\Longleftrightarrow\qquad 
q_c := \Pr[ c = j ]
= 
\frac{\| \bv{A}\bv{e}_j \|^2}{\|\bv{A}\|_\F^2} 
,\quad j=1,\ldots, d.
\end{equation*}
\item For each $r$, we can sample a column from the entries of the row $\bv{A}^\T\bv{e}_r$ proportional to the squared magnitude distribution $\mathcal{D}_r^{\textup{col}}$:
\begin{equation*}
c\sim  \mathcal{D}_r^{\textup{col}}
\qquad\Longleftrightarrow\qquad 
\Pr[ c = j ]
= 
\frac{| \bv{e}_r^\T \bv{A} \bv{e}_j|^2}{\|\bv{e}_r^\T\bv{A}\|^2}
,\quad j=1,\ldots, d.
\end{equation*}
\end{itemize}
\end{assumption}

We assume sampling/querying can be done at unit cost for ease of exposition.\footnote{
%
After all, we (and most other papers in the area) already abstract away arguably more important costs such as the bit complexity required for the algorithms to work in finite arithmetic!
}
Of course, building such a sampler from scratch requires at least $\operatorname{nnz}(\bv{A})$ time.
For a more detailed overview of data-structure for sample and query access which are analogous to those implementable with quantum RAM (QRAM), see \cite{tang_19,tang_23}.

We note that \cref{asm:sample_query} is only a mild generalization of what is typically assumed when analyzing randomized Kaczmarz and related algorithms \cite{strohmer_vershynin_08}. 
As such, we believe that in most situations where row-norm sampling is reasonable, the sample/query access \cref{asm:sample_query} is also reasonable.
Alternately, as with randomized Kaczmarz type algorithms, one can use uniform sampling.
The same general analysis approach as we take in this paper can be easily extended to uniform sampling by introducing a parameter quantifying the distance of the row-norms from being uniform \cite{needell_srebro_ward_14}. 

\subsubsection{Sampling from a compressed representation}
\label{sec:compressed_representation}

In order to solve \cref{problem:sampling}, we follow \cite{gilyn_song_tang_22,shao_montanaro_22,shao_montanaro_23} and  produce an approximation of the form  $\bv{x} = \bv{A}^\T\bv{y}$, where $\bv{y}$ is some sparse vector that we efficiently maintain.
Clearly, we can query an entry of $\bv{x}$ efficiently by simply reading the relevant entries of the relevant rows of $\bv{A}$.
The following guarantees that we can also sample from $\bv{x}$.

\begin{imptheorem}[\protect{Proposition 4.3 in \cite{tang_19}}]
\label{thm:sampling}
Instantiate \cref{asm:sample_query}.
Given a vector $\bv{y}\in\mathbb{R}^n$, define
\begin{equation}
    \phi(\bv{y}) := \operatorname{nnz}(\bv{y})\sum_{r=1}^{n}  \frac{|\bv{y}_r|^2\|\bv{e}_r^\T\bv{A}\|^2}{\|\bv{A}^\T\bv{y}\|^2}.
\end{equation}
There is an algorithm that produces a sample from $\mathcal{D}_{\bv{x}}$, where $\bv{x} = \bv{A}^\T\bv{y}$, using an expected $O(\operatorname{nnz}(\bv{y})\phi(\bv{y}))$ cost.
\end{imptheorem}

The quantity $\phi(\bv{y})$ can be viewed as a measure of how much large entries of the relevant rows of $\bv{A}$ correspond to large entries of $\bv{x} = \bv{A}^\T\bv{y}$, which intuitively is relevant for sampling from $\bv{x}$ given $\bv{y}$ and a way to sample from $\bv{A}$.
Indeed, while large entries of $\bv{x}$ must come from a combination of large entries of $\bv{y}$ and large entries in the corresponding rows of $\bv{A}$, due to cancellation, small entries of $\bv{x}$ can also arise from the combination of a large entries of $\bv{y}$ and  large entries in the corresponding row of $\bv{A}$.
We provide a detailed explication of the approach in \cref{sec:vector_sampling}.

\subsection{Past work}


A number of quantum-inspired algorithms for \cref{problem:sampling} and adjacent problems have been developed based on varying techniques \cite[etc.]{chia_gilyen_andraspal_li_lin_tang_wang_22,chepurko_clarkson_horesh_lin_woodruff_22,bakshi_tang_24}.
Most related to the present work is a sequence of papers which aim to approximate gradient descent type algorithms \cite{gilyn_song_tang_22,shao_montanaro_22,shao_montanaro_23}. 
Specifically, these algorithms maintain a sparse iterate $\bv{y}_k$ in the dual space such that $\E[\bv{A}^\T\bv{y}_k]$ is equal to a gradient decent algorithm.
In contrast to standard stochastic gradient decent/randomized Kacmzarz type iterations which use a row-submatrix of $\bv{A}$, each iteration of the algorithms in \cite{gilyn_song_tang_22,shao_montanaro_22,shao_montanaro_23} is performed by sampling both rows and columns of $\bv{A}$.

In \cite{gilyn_song_tang_22}, an algorithm for a more general problem of ridge-regression is introduced. 
For \cref{problem:sampling} studied in the present paper, the algorithm of \cite{gilyn_song_tang_22} builds the representation of the solution in $\widetilde{O}(\condF^6\cond^2 / \varepsilon^4)$ time.
Samples can be subsequently generated in $\widetilde{O}(\condF^6\cond^6 / \varepsilon^4)$ time.
The analysis is simple and mostly non-asymptotic with explicit constants.
With regards to simplicity and elegance, we identify only two minor shortcomings.
First, the algorithm requires that $\bv{b}$ is explicitly sparsified as a pre-processing step, which is somewhat unsatisfying.
Second, the iterate $\bv{A}^\T\bv{y}_k$, conditioned on the row-samples used, is not equivalent to the widely studied randomized Kaczmarz method \cite{strohmer_vershynin_08}.
As such, the analysis follows a much more general analysis of stochastic gradient algorithms with stochastic gradients satisfying an additive error bound (whereas the stochastic gradients used by randomized Kaczmarz enjoy stronger multiplicative guarantees for consistent systems).

Subsequently, \cite{shao_montanaro_22} introduce an algorithm that is conditionally equivalent to randomized Kaczmarz with averaging \cite{moorman_tu_molitor_needell_20}, and hence, in principle, admits a simple analysis. 
The algorithm, which samples multiple rows each iteration, is claimed to achieve an improved construction time $\widetilde{O}(\condF^6\cond^2 / \varepsilon^2)$ and $\widetilde{O}(\condF^6\cond^4)$ sample time.\footnote{In \cite[footnote 10]{gilyn_song_tang_22}, it is claimed that their analysis can recover a similar result to \cite{shao_montanaro_22} in the case of consistent systems.}
Subsequently, in an update \cite{shao_montanaro_23}, these rates are improved to $\widetilde{O}(\condF^4\cond^2 / \varepsilon^2)$ and $\widetilde{O}(\condF^4\cond^2)$ respectively.
However, the analysis in \cite{shao_montanaro_22,shao_montanaro_23} relies heavily on big-$O$ notation and is difficult to follow since many important quantities are not clearly defined.
Moreover, as we expand on in \cref{sec:SM_error}, we were unable to verify several key steps in \cite{shao_montanaro_22,shao_montanaro_23}. 
In particular, we believe (i) the construction cost is actually $O(\condF^6/\varepsilon^2)$ and (ii) the analysis of the sample cost is incorrect, as it relies heavily on a ``fact'' about norms which we show is false.

Thus, to the best of our knowledge, the following question remains open:
\vspace{-1em}
\begin{quote}
\textbf{Question:}~
\textit{Is there a simple gradient-based algorithm that can easily be shown to solve \cref{problem:sampling} with both a quartic dependence on $\condF$ and quadratic dependence on $1/\varepsilon$?}
\end{quote}




\subsection{Main result and technical overview}
\label{sec:results}
The aim of the present work is to provide an \emph{elementary, non-asymptotic, and fully self-contained} proof that a simple gradient-based algorithm (\cref{alg:main}) solves \cref{problem:sampling} with
\begin{equation}
\label{eqn:rates}
{O}\left(\frac{\condF^4\cond^2}{\varepsilon^2}\right)
\text{ construction,}\quad
\widetilde{O}\left( \condF^2 \right) 
\text { query,}
\quad\text{and}\quad
\widetilde{O}\left( \condF^4 \cond^6 \right) \text{ sample cost}.
\end{equation}

\begin{algorithm}[ht!]
\caption{Regression via row and column sub-sampled stochastic gradient descent}
\label{alg:main}
\begin{algorithmic}[1]
\Require sample/query access to $\bv{A}$ and $\bv{b}$, step size $\alpha$, row/column batch sizes $R$, $C$, iterations $K$
\State set $\bv{y} = (y_1, \ldots, y_n) = \bv{0}$
\For{$k=1,2,\ldots, K$} \Comment{Update $\bv{y}$ from $\bv{y}_k$ to $\bv{y}_{k+1}$}
\State sample $C$ column indices $\mathcal{C}_k \sim (\mathcal{D}^{\textup{col}})^C$ 
\For{$c\in\mathcal{C}_k$}
\State compute $\displaystyle\beta_c := \sum_{i:  y_i\neq 0} A_{i,c} y_i$  
\Comment{$(\bv{e}_c^\T \bv{A}^\T)\bv{y}$}
\EndFor
\State sample $R$ row indices $\mathcal{R}_k \sim (\mathcal{D}^{\textup{row}})^R$
\For{$r\in\mathcal{R}_k$}
\State compute $\displaystyle\gamma_r := \left(\frac{1}{C}\sum_{c\in\mathcal{C}_{k}} (q_{c})^{-1} A_{r,c} \beta_c  \right) - b_r$
\Comment{$\displaystyle\left(\frac{1}{C}\sum_{c\in\mathcal{C}_{k}} (q_{c})^{-1} \bv{e}_r^\T \bv{A} \bv{e}_{c} \beta_c  \right) - \bv{e}_r^\T \bv{b}$}
\State update $\displaystyle y_r \gets y_r - \alpha \frac{\gamma_r}{R p_r} $
\EndFor
\EndFor
\Ensure compressed representation $\bv{y}$ of approximate solution $\bv{x} = \bv{A}^\T\bv{y}$
\end{algorithmic}
\end{algorithm}

We follow the same general strategy of \cite{gilyn_song_tang_22,shao_montanaro_22,shao_montanaro_23}, and analyze a stochastic gradient algorithm (\cref{alg:main}) for producing a compressed representation $\bv{y}$ of an approximate solution $\bv{x} = \bv{A}^\T\bv{y}$ to \cref{eqn:linear_system} that admits efficient sample and query access.
The main pieces we must analyze to obtain the bounds \cref{eqn:rates} are (i) the cost to construct $\bv{y}$, (ii) the accuracy of $\bv{x}$, and (iii) the cost to sample and query from $\bv{x}$.
In the remainder of this section, we provide an overview of our analysis of these three aspects.

As with \cite{shao_montanaro_22,shao_montanaro_23}, we make use of multiple row and column samples per iteration. 
However, unlike these past works, we reuse the same column samples for all rows.
As we note in \cref{rem:col_sampling}, this algorithmic advancement substantially decreases the running time of our algorithm.
Our focus is simplicity, so we do not consider the case that $\bv{b}\not\in\range(\bv{A})$ nor a more general problem like the ridge regression problem. 
We expect the ideas in this paper can be extended to these settings.

\paragraph{Construction.}
At each iteration $k$, we sample $R$ row indices and $C$ column indices, both independently with replacement from $\mathcal{D}^{\textup{row}}$ and $\mathcal{D}^{\textup{col}}$ respectively. 
The ordered lists $\mathcal{R}_k$ and $\mathcal{C}_k$ contain these indices. 
Given a step size parameter $\alpha$, we then produce an update
\begin{equation}
    \label{eqn:yk_update}
    \bv{y}_{k+1} = \bv{y}_k - \alpha \left( \frac{1}{R}\sum_{r\in\mathcal{R}_k} \frac{1}{p_{r}} \bv{e}_{r} \left(\left(\frac{1}{C}\sum_{c\in\mathcal{C}_{k}} \frac{1}{q_{c}} (\bv{e}_r^\T\bv{A} \bv{e}_{c}) (\bv{e}_c^\T \bv{A}^\T \bv{y}_k) \right) - \bv{e}_r^\T\bv{b} \right) \right),
    \qquad 
    \bv{y}_0 = \bv{0},
\end{equation}
where 
\begin{equation}
    p_r := \frac{\|\bv{e}_r^\T\bv{A}\|^2}{\|\bv{A}\|_\F^2}
    ,\qquad 
    q_c := \frac{\|\bv{A}\bv{e}_c\|^2}{\|\bv{A}\|_\F^2}
    .
\end{equation}
We will use the parameters:
\begin{equation}
\label{eqn:RCk_set}
\alpha = \frac{1}{\|\bv{A}\|^2}
,\qquad
R = \left\lceil \frac{2\condF^2}{\cond^2}\right\rceil,
\qquad
C = \left\lceil  \frac{10\condF^2}{\varepsilon^2}\right\rceil
,\qquad
K = \left\lceil 4 \cond^2 \cdot \log\left(\frac{1}{\varepsilon}\right)\right\rceil .
\end{equation}

This recurrence is efficient to maintain:
\begin{lemma}\label{thm:y_cost}
Suppose $\varepsilon<1/4$ and run \cref{alg:main} with parameters as in \cref{eqn:RCk_set} to obtain $\bv{y}$.
Then $\bv{y}$ has at most $8\condF^2 \log(1/\varepsilon)$ nonzero entries, and can be generated using $\widetilde{O}(\condF^4 \cond^2/\varepsilon^2)$ cost.
\end{lemma}

\begin{proof}
At each iteration $k$, we update at most $R$ nonzero entries of our iterate (corresponding to the indices in $\mathcal{R}_k$). 
Hence, $\bv{y} = \bv{y}_{K+1}$ has at most $s = KR$ non-zero entries.

Each iteration requires computing $C$ inner products $\bv{e}_c^\T\bv{A}^\T\bv{y}$.
Since $\bv{y}$ is $s$-sparse, each requires $O(s)$ cost.
Subsequently, we require computing $R$ coefficients $\gamma_r$ each at a cost of $O(C)$.
Finally, we update $R$ entries of $\bv{y}$.
Thus, the cost at iteration $k$ is $O(RC+sC) = O(kRC)$ and the overall time up to iteration $K$ is $O(K^2RC)$.
\end{proof}

\begin{remark}\label{rem:col_sampling}
If we used different column indices for each row $r$ (as done in past work \cite{shao_montanaro_22,shao_montanaro_23}), then the cost per iteration would be higher by a factor of $R$, since there would now be $RC$ inner products.
\end{remark}

\paragraph{Accuracy.}
We must bound the sample sizes $R$ and $C$ and the iterations $K$ required to ensure that $\bv{x}$ is close to $\bv{x}^* = \bv{A}^+\bv{b}$.
The recurrence \cref{eqn:yk_update} gives updates
\begin{equation}
\label{eqn:xk_update}
\bv{x}_{k+1} = \bv{x}_k - \alpha \left( \frac{1}{R}\sum_{r\in\mathcal{R}_k} \frac{1}{p_{r}} \bv{A}^\T\bv{e}_{r} \left(\left(\frac{1}{C}\sum_{c\in\mathcal{C}_{k}} \frac{1}{q_{c}} (\bv{e}_r^\T\bv{A}\bv{e}_c)( \bv{e}_c^\T \bv{x}_k) \right) - \bv{e}_r^\T\bv{b} \right) \right),\qquad \bv{x}_0 = \bv{0}.
\end{equation}
Observe that $\E[\bv{x}_{k+1} | \bv{x}_k] = \bv{x}_k - \alpha(\bv{A} \bv{x}_k - \bv{b})$.
Hence \cref{eqn:yk_update} can be viewed as a compressed representation of a gradient descent type algorithm for the objective $\bv{x}\mapsto \|\bv{b} - \bv{A}\bv{x}\|^2$.
Moreover, $\E[\bv{x}_{k+1} | \mathcal{R}_k, \bv{x}_k] = \bv{x}_k - \alpha\sum_{r\in\mathcal{R}_k} \bv{A}^\T\bv{e}_r(\bv{e}_r^\T\bv{A}  \bv{x}_k - \bv{b}_r)$, which is randomized Kaczmarz with averaging \cite{moorman_tu_molitor_needell_20}.
In \cref{sec:x_bound} we follow the analysis of such methods in order to analyze the accuracy of $\bv{x}_{k+1}$.

\begin{restatable}{theorem}{mainx}
\label{thm:main_x}
Suppose $\varepsilon<1/4$ and run \cref{alg:main} with parameters as in \cref{eqn:RCk_set} to obtain $\bv{y}$.
Let $\bv{x} = \bv{A}^\T\bv{y}$. 
Then,
\begin{equation*}
\E\left[ \| \bv{x}^* - \bv{x} \|^2 \right]
\leq 2\varepsilon^2  \| \bv{x}^* \|^2.
\end{equation*}
\end{restatable}
\paragraph{Querying and sampling.}
Since $\bv{y}$ is sparse, we can easily evaluate entries of $\bv{x} = \bv{A}^\T\bv{y}$.
In \cref{sec:y_bound}, we show that $\phi(\bv{y})$ is not too large, which allows efficient sampling from $\bv{x}$ by \cref{thm:sampling}.
The approach here is to show that the entries of $\bv{y}$ can be bounded in terms of $\|\bv{x}\| = \|\bv{A}^\T\bv{y}\|$.
Intuitively, we might expect this since $\bv{y}$ does not change too much from iteration to iteration, as can be seen in \eqref{eqn:yk_update}.

\begin{restatable}{theorem}{mainy}\label{thm:main_y}
Suppose $\varepsilon<1/4$ and run \cref{alg:main} with parameters as in \cref{eqn:RCk_set} to obtain $\bv{y}$.
Then, with probability at least $9/10$,
\begin{equation*}
    \phi(\bv{y}) \leq \const\cdot \condF^2 \cond^6 \log\left(\frac{1}{\varepsilon}\right)^3.
\end{equation*}
\end{restatable}

\paragraph{Putting it together.}
Combing the above results gives \cref{eqn:rates}.
Specifically, \cref{thm:main_x} implies the accuracy part of the result.
\Cref{thm:y_cost} gives the sparsity guarantee.
Finally, \cref{thm:main_y}, when combined with the sparsity bound and \cref{thm:sampling}, gives the sampling cost.

\section{Analysis}

In \cref{sec:x_bound} we analyze the accuracy of the iterate $\bv{x}_{K+1} = \bv{A}^\T\bv{y}_{K+1}$. 
Then, in \cref{sec:y_bound}, we analyze how efficiently we can sample from the approximate solution.
Together, these give the rates described in \cref{sec:results}.
A number of elementary calculations are deferred to \cref{sec:deferred_proofs}.

\subsection{Analysis of accuracy}
\label{sec:x_bound}

Our first main technical task is to analyze $\|\bv{x}^* -\bv{x}_{K+1} \|$.
Define the random variables 
\begin{equation}
\label{eqn:RV_defs}
\bv{M}_k := \frac{1}{R}\sum_{r\in\mathcal{R}_k} \frac{1}{p_{r}} \bv{A}^\T \bv{e}_{r} \bv{e}_{r}^\T\bv{A}
,\qquad
\bv{z}_k := \frac{1}{R}\sum_{r\in\mathcal{R}_k} \frac{1}{p_{r}} \bv{A}^\T \bv{e}_{r} \left(\frac{1}{C}\sum_{c\in\mathcal{C}_{k}} \frac{1}{q_{c}} \bv{e}_r^\T\bv{A} \bv{e}_{c} \bv{e}_c^\T \bv{x}_k\right).
\end{equation}
Since $\bv{b}\in\range(\bv{A})$, then $\bv{b} = \bv{A}\bv{x}^*$.
Therefore, \cref{eqn:xk_update} becomes
\begin{equation}
    \bv{x}_{k+1} 
    = \bv{x}_k - \alpha (\bv{z}_k - \bv{M}_k\bv{x}^*)
    = \bv{x}_k - \alpha \bv{M}_k (\bv{x}_k - \bv{x}^*) + \alpha(\bv{M}_k\bv{x}_k - \bv{z}_k).
\end{equation}
Finally, subtracting $\bv{x}^*$ from both sides we obtain a recurrence
\begin{equation}
    \bv{x}_{k+1} - \bv{x}^* 
    = (\bv{I}-\alpha \bv{M}_k)(\bv{x}_k - \bv{x}^*) + \alpha(\bv{M}_k\bv{x}_k - \bv{z}_k). \label{eqn:error_recurrence}
\end{equation}
We will iteratively bound the error $\|\bv{x}_{k+1} - \bv{x}^*\|$.

\subsubsection{Preliminary Lemmas}

We begin by stating a few basic facts that we prove in \cref{sec:deferred_proofs:accuracy}.
\begin{restatable}{lemma}{Mkbd}
\label{thm:Mk_bd}
Let $\bv{M}_k$ be defined as in \cref{eqn:RV_defs}.
For any $R>0$, $\E[\bv{M}_k] = \bv{A}^\T\bv{A}$ and
\begin{equation*}
    \E\left[\bv{M}_k^2\right]
    = \frac{1}{R}\|\bv{A}\|_\F^2 \bv{A}^\T\bv{A} + \left(1 - \frac{1}{R} \right) (\bv{A}^\T\bv{A})^2.
\end{equation*}
\end{restatable}

\begin{restatable}{lemma}{colboundsingle}
\label{thm:col_bound_single}
Let $\bv{M}_k$ and $\bv{z}_k$ be defined as in \cref{eqn:RV_defs}.
    For any $R,C>0$, $\E[\bv{M}_k \bv{x}_k - \bv{z}_k | \bv{M}_k,\bv{x}_k] = \bv{0}$ and
    \begin{equation*}
    \E\left[ \| \bv{M}_k \bv{x}_k - \bv{z}_k \|^2 \middle| \bv{x}_k \right] 
    \leq \left( \frac{\|\bv{A}\|_\F^4}{RC} + \frac{\|\bv{A}\|_\F^2\|\bv{A}\|^2 }{C} \right)\|\bv{x}_k \|^2.
    \end{equation*}
\end{restatable}

\subsubsection{Main bound}

The proof of our main bound is now straightforward, and follows a standard analysis approach for stochastic gradient methods.

\mainx*

\begin{proof}
We will prove
\begin{equation}
\E\left[ \|\bv{x}_{K+1} - \bv{x}^*\|^2 \right]
\leq \left( \rho^{K+1} + \varepsilon^2 \right) \| \bv{x}^* \|^2
,\qquad
\rho := 1 - \frac{1}{2\cond^2}.
\end{equation}
Since $K\geq 4 \kappa^2 \log(1/\varepsilon)$, it holds that $\rho^{K+1} \leq \varepsilon^2$, which then gives the result.

Define $\bm{\delta}\bv{x}_j = \bv{x}_j - \bv{x}^*$.
With this notation, \cref{eqn:error_recurrence} becomes
\begin{equation}
    \bm{\delta}\bv{x}_{k+1} = (\bv{I} - \alpha \bv{M}_k)\bm{\delta}\bv{x}_k + \alpha(\bv{M}_k \bv{x}_k - \bv{z}_k).
\end{equation}
By the law of iterated expectation and \cref{thm:col_bound_single},
\begin{align}
\E\left[(\bv{M}_k \bv{x}_k - \bv{z}_k)^\T(\bv{I} - \alpha \bv{M}_k)\bm{\delta}\bv{x}_k\right]
&= \E\left[\E\left[(\bv{M}_k \bv{x}_k - \bv{z}_k)^\T(\bv{I} - \alpha \bv{M}_k)\bm{\delta}\bv{x}_k\middle|\bv{M}_k,\bv{x}_k\right]\right]
\\&= \E\left[\E\left[\bv{M}_k \bv{x}_k - \bv{z}_k\middle|\bv{M}_k,\bv{x}_k\right]^\T(\bv{I} - \alpha \bv{M}_k)\bm{\delta}\bv{x}_k\right] = \bv{0}.
\end{align}
Therefore, expanding $\|\bm{\delta}\bv{x}_{k+1}\|^2$,
\begin{align}
    \E\left[\| \bm{\delta}\bv{x}_{k+1}\|^2\right] 
    &= \E\left[\| (\bv{I} - \alpha \bv{M}_k)\bm{\delta}\bv{x}_k\|^2 + 2\alpha(\bv{M}_k \bv{x}_k - \bv{z}_k)^\T(\bv{I} - \alpha \bv{M}_k)\bm{\delta}\bv{x}_k + \alpha^2\|\bv{M}_k \bv{x}_k - \bv{z}_k \|^2\right]
    \\&=\underbrace{\E\left[\| (\bv{I} - \alpha \bv{M}_k)\bm{\delta}\bv{x}_k\|^2\right]}_{\text{term 1}} + \underbrace{\alpha^2\E\left[\|\bv{M}_k \bv{x}_k - \bv{z}_k \|^2\right]}_{\text{term 2}}.\label{eqn:ek1_bound}
\end{align}
We will now bound these two terms.

\paragraph{Bounding term 1.}
We begin with the first term in \cref{eqn:ek1_bound}.
Let $\bm{\Pi}$ be the orthogonal projector onto $\range(\bv{A}^\T)$.
Observe that $\bv{x}_k  = \bv{A}^\T\bv{y}_k \in \range(\bv{A}^\T)$ and $\bv{x}^* = \bv{A}^+ \bv{b} \in \range(\bv{A}^\T)$.
Therefore $\bm{\delta}\bv{x}_k\in\range(\bv{A}^\T)$ so $\bm{\delta}\bv{x}_k = \bm{\Pi}\bm{\delta}\bv{x}_k$.

Now,
\begin{align}
\E\left[\| (\bv{I} - \alpha \bv{M}_k)\bm{\delta}\bv{x}_k\|^2 \middle| \bv{x}_k \right]
&= \bm{\delta}\bv{x}_k^\T \E\left[(\bv{I} - \alpha \bv{M}_k)^2\right]\bm{\delta}\bv{x}_k
\\&= \bm{\delta}\bv{x}_k^\T \bm{\Pi} \E\left[(\bv{I} - \alpha \bv{M}_k)^2\right]\bm{\Pi}\bm{\delta}\bv{x}_k
\\&\leq \left\|\bm{\Pi} \E \left[(\bv{I} - \alpha \bv{M}_k)^2 \right] \bm{\Pi}\right\|\|\bm{\delta}\bv{x}_k\|^2.
\end{align}
Using \cref{thm:Mk_bd} we find that
\begin{align}
\E\left[ (\bv{I} - \alpha \bv{M}_k)^2\right]
&= \E\left[\bv{I} - 2\alpha \bv{M}_k + \alpha^2 \bv{M}_k^2\right]
\\&= \bv{I}- 2\alpha \bv{A}^\T\bv{A} + \alpha^2 (R^{-1}\|\bv{A}\|_\F^2\bv{A}^\T\bv{A} + (1-R^{-1})(\bv{A}^\T\bv{A})^2)
\\&\preceq\bv{I} - 2\alpha \bv{A}^\T\bv{A}  + \alpha^2 (\bv{A}^\T\bv{A})^2+ \alpha^2 R^{-1}\|\bv{A}\|_\F^2\bv{A}^\T\bv{A}. 
\\&= (\bv{I} - \alpha \bv{A}^\T\bv{A})^2 + \alpha^2 R^{-1}\|\bv{A}\|_\F^2\bv{A}^\T\bv{A} 
\\&\preceq \bv{I} - \alpha \bv{A}^\T\bv{A} + \alpha^2 R^{-1}\|\bv{A}\|_\F^2\bv{A}^\T\bv{A},
\end{align}
where we have used that $-\alpha^2 R^{-1}(\bv{A}^\T\bv{A})^2 \preceq \bv{0}$ and $(\bv{I} - \alpha \bv{A}^\T\bv{A})^2 \preceq \bv{I} - \alpha \bv{A}^\T\bv{A}$ (all nonzero eigenvalues of $\bv{A}^\T\bv{A}$ are less than $1/\alpha$ in magnitude).

Therefore, if $R \geq 2\alpha\|\bv{A}\|_\F^2 = 2\condF^2/\cond^2$, then
\begin{equation}
\label{eqn:ek1_bound_term1}
\left\| \bm{\Pi}\E \left[(\bv{I} - \alpha \bv{M}_k)^2\right]\bm{\Pi} \right\|
\leq \| \bm{\Pi}(\bv{I} - (\alpha/2) \bv{A}^\T\bv{A})\bm{\Pi}\|
= \max_{\substack{\lambda \in\operatorname{spec}(\bv{A}^\T\bv{A})\\\lambda\neq0}} | 1 - (\alpha/2) \lambda |
= \rho,
\end{equation}
where the final equality follows because $\alpha /2 > \lambda$ for each $\lambda \in\operatorname{spec}(\bv{A}^\T\bv{A})$.

\paragraph{Bounding term 2.}
We now bound the second term in \cref{eqn:ek1_bound}.
By \cref{thm:col_bound_single} and the assumption $\alpha = 1/\|\bv{A}\|^2$, if $C \geq \alpha \|\bv{A}\|_\F^2/\gamma^2 = \alpha^2 \|\bv{A}\|_\F^2\|\bv{A}\|^2/\gamma^2 \geq \max\big\{ \alpha^2 \|\bv{A}\|_\F^4/(R\gamma^2) ,\alpha^2 \|\bv{A}\|_\F^2\|\bv{A}\|^2 / \gamma^2 \big\}$, then
\begin{equation}\label{eqn:ek1_bound_term2}
    \alpha^2 \E\left[ \| \bv{M}_j \bv{x}_j - \bv{z}_j \|^2 \big| \bv{x}_j \right] 
    \leq \alpha^2 \left( \frac{\|\bv{A}\|_\F^4}{RC} + \frac{\|\bv{A}\|_\F^2\|\bv{A}\|^2 }{C} \right)\|\bv{x}_j \|^2
    \leq 2 \gamma^2 \| \bv{x}_j\|^2.
\end{equation}
\paragraph{Finishing up.}
Plugging \cref{eqn:ek1_bound_term1,eqn:ek1_bound_term2} into \cref{eqn:ek1_bound} and using the law of iterated expectations we arrive at an expression
\begin{equation}\label{eqn:ek1_xj}
\E\left[\| \bm{\delta}\bv{x}_{k+1}\|^2\right] 
\leq  \rho \E\left[\| \bm{\delta}\bv{x}_{k}\|^2\right] + 2 \gamma^2 \E\left[ \|\bv{x}_j \|^2 \right].
\end{equation}

Our aim is now to relate the quantities in \cref{eqn:ek1_xj} to $\|\bv{x}^*\|$.
Clearly $\E\|\bv{x}_0\|^2 = 0 \leq \|\bv{x}^*\|^2$.
Now, assume, for the sake of induction, that for each $j\leq k$,
\begin{equation}
\E\left[\| \bv{x}_{j}  \|^2\right]
\leq  2 \| \bv{x}^*\|^2.
\end{equation}
Then, \cref{eqn:ek1_xj} yields a bound
\begin{equation}
\E\left[\| \bm{\delta}\bv{x}_{k+1}\|^2\right] 
\leq  \rho \E\left[\| \bm{\delta}\bv{x}_{k}\|^2\right] + 4 \gamma^2 \|\bv{x}^* \|^2 .
\end{equation}
Unrolling this, and using that $\|\bm{\delta}\bv{x}_0\| = \|\bv{x}_0 - \bv{x}^*\| = \|\bv{x}^*\|$, we find that 
\begin{align}
\E\left[\| \bm{\delta}\bv{x}_{k+1}\|^2\right] 
&\leq  \rho^{k+1} \| \bv{x}^*\|^2 + 4 \gamma^2 \|\bv{x}^* \|^2 \sum_{j=0}^{k} \rho^{j}
\\&\leq  \rho^{k+1} \| \bv{x}^*\|^2 + 4 \gamma^2 \|\bv{x}^* \|^2 \sum_{j=0}^{\infty} \rho^{j}
\\&= \rho^{k+1} \| \bv{x}^*\|^2 + 4 (1-\rho)^{-1} \gamma^2 \|\bv{x}^* \|^2.
\end{align}
Since $(1-\rho)^{-1} = 2\cond^2$, setting $\varepsilon^2 = 8 \cond^2\gamma^2$ (so that $C \geq (\condF^2/\cond^2) / (\varepsilon^2 / 10\cond^2) = 10 \condF^2/\varepsilon^2$) gives that 
\begin{equation}
\E\left[ \|\bv{x}_{k+1} - \bv{x}^*\|^2 \right]
\leq \left(\rho^{k+1} + \varepsilon^2 \right) \| \bv{x}^* \|^2
\leq 2 \varepsilon^2 \|\bv{x}^*\|^2.
\label{eqn:xk_intermediate}
\end{equation}
Thus, by the triangle inequality, Jensen's inequality, and \cref{eqn:xk_intermediate}
\begin{align}
\E\left[\|\bv{x}_k\|^2\right]  
&\leq 
\E\left[(\| \bv{x}^* - \bv{x}_k \|+\|\bv{x}^*\|)^2\right] 
\\&\leq \E\left[\| \bv{x}^* - \bv{x}_k \|^2\right] + 2  \|\bv{x}^*\| \sqrt{\E\left[\| \bv{x}^* - \bv{x}_k \|^2\right]} + \|\bv{x}^*\|^2 
\\&\leq 2 \varepsilon^2 \|\bv{x}^*\|^2 + 2 \sqrt{2} \varepsilon \|\bv{x}^*\|^2 + \|\bv{x}^*\|^2
< 2\|\bv{x}^*\|^2,
\label{eqn:xk_bound}
\end{align}
where we have used that $\varepsilon < 1/4$.
The inductive hypothesis is satisfied and the result follows.
\end{proof}

\subsection{Analysis of sampling}
\label{sec:y_bound}

To apply \cref{thm:sampling}, we must bound
\begin{equation}
    \phi(\bv{y}_{k+1}) = \sum_{i}  \frac{|(\bv{y}_{k+1})_i|^2\| \bv{a}_i \|^2}{\|\bv{A}^\T\bv{y}_{k+1}\|^2}.
\end{equation}
Towards this end, observe that 
\begin{equation}\label{eqn:D_def}
    \sum_i \|\bv{a}_i \|^2 |(\bv{y}_{k+1})_i|^2
= \|\bv{D} \bv{y}_{k+1}\|^2
,\quad \bv{D} := \operatorname{diag}(\|\bv{e}_1^\T\bv{A}\|, ... ,\|\bv{e}_n^\T\bv{A}\|).
\end{equation}
Define random variables
\begin{equation}\label{eqn:uv_def}
    \bv{u}_{k} := \frac{1}{R} \sum_{r\in\mathcal{R}_k} \frac{1}{p_r}\bv{e}_r \left(\frac{1}{C}\sum_{c\in\mathcal{C}_{k}} \frac{1}{q_{c}} (\bv{e}_r^\T\bv{A}\bv{e}_{c})( \bv{e}_c^\T \bv{x}_k) \right) 
    ,\qquad
    \bv{v}_k := \frac{1}{R} \sum_{r\in\mathcal{R}_k} \frac{1}{p_r}\bv{e}_r \bv{e}_r^\T\bv{b}.
\end{equation}
The recurrence \cref{eqn:yk_update} then gives
\begin{equation}\label{eqn:Dyk1_def}
    \bv{D}\bv{y}_{k+1} = \bv{D}\bv{y}_k - \alpha \bv{D}(\bv{u}_k-\bv{v}_k).
\end{equation}
Our general approach will be to show that $\alpha\| \bv{D}(\bv{u}_k-\bv{v}_k) \|$ is not to large, so that $\|\bv{D}\bv{y}_{k+1}\| \approx \|\bv{D}\bv{y}_{k}\|$.

\subsubsection{Preliminary Lemmas}

We begin by stating a few basic facts that we prove in \cref{sec:deferred_proofs:sampling}.

\begin{restatable}{lemma}{Dgtermabound}
\label{thm:Dg_term1_bound}
Let $\bv{D}$ be defined as in \cref{eqn:D_def} and $\bv{u}_k$ be defined as in \cref{eqn:uv_def}.
For any $R,C>0$, 
\begin{equation*}
    \E\left[ \left\| \bv{D}\bv{u}_k \right\|^2 \middle| \bv{x}_k \right]
    \leq \left(\frac{\|\bv{A}\|_\F^4}{RC} + \frac{\|\bv{A}\|_\F^2\|\bv{A}\|^2}{R}  +  \| \bv{A}\|^4\right) \|\bv{x}_k \|^2.
\end{equation*}
\end{restatable}

\begin{restatable}{lemma}{Dgtermbbound}
\label{thm:Dg_term2_bound}
Let $\bv{D}$ be defined as in \cref{eqn:D_def} and $\bv{v}_k$ be defined as in \cref{eqn:uv_def}.
For any $R>0$,
\begin{equation*}
    \E\left[ \left\| \bv{D}\bv{v}_k \right\|^2 \middle| \bv{x}_k \right]
    \leq \frac{\|\bv{A}\|_\F^2 \|\bv{b}\|^2}{R} +  \| \bv{A}\|^2 \|\bv{b} \|^2.
\end{equation*}
\end{restatable}

Together, these give a bound on $\|\bv{D}(\bv{u}_k - \bv{v}_k) \|$.
\begin{lemma}
\label{thm:Dg_bound}
Let $\bv{D}$ be defined as in \cref{eqn:D_def} and $\bv{u}_k$ and $\bv{v}_k$ be defined as in \cref{eqn:uv_def}.
Suppose $\alpha$, $R$, and $C$ are as in \cref{thm:main_x}. 
Then
    \begin{equation*}
    \alpha^2\E\left[\|\bv{D} (\bv{u}_k - \bv{v}_k) \|^2 \right]
    \leq 10 \cond^4 \|\bv{x}^*\|^2.
    \end{equation*}
\end{lemma}

\begin{proof}
Recall $\bv{x}_j = \bv{A}^\T\bv{y}_j$ and note that $\bv{D}\bv{e}_r = \|\bv{e}_r^\T\bv{A}\|\bv{e}_r$.
Then using that $(x+y)^2 \leq 2(x^2+y^2)$,
\begin{equation}
    \E\left[\|\bv{D} (\bv{u}_j-\bv{v}_k)\|^2\middle| \bv{x}_k \right]
    \leq 2\E\left[ \left\| \bv{D}\bv{u}_k \right\|^2 \middle| \bv{x}_k \right] + 2\E\left[ \left\| \bv{D}\bv{v}_k \right\|^2\right].
    \label{eqn:Dg_two_terms}
\end{equation}

Plugging \cref{thm:Dg_term1_bound,thm:Dg_term2_bound} into \cref{eqn:Dg_two_terms} we obtain a bound
\begin{align}
    \E\left[\|\bv{D} (\bv{u}_k-\bv{v}_k)\|^2\middle| \bv{x}_k \right]
    \nonumber&\leq 2\left(\frac{\|\bv{A}\|_\F^4}{RC} + \frac{\|\bv{A}\|_\F^2\|\bv{A}\|^2}{R}  +  \| \bv{A}\|^4\right) \|\bv{x}_k \|^2 
    \\&\hspace{8em} + 2\left( \frac{\|\bv{A}\|_\F^2\|\bv{b}\|^2}{R} +  \| \bv{A}\|^2 \|\bv{b} \|^2 \right).
\end{align}
By \cref{eqn:xk_bound}, $\E\left[\|\bv{x}_k\|^2\right]  \leq 2\|\bv{x}^*\|^2$.
\cref{thm:main_x}.
Moreover, $\|\bv{b}\| =\|\bv{A}\bv{x}^*\| \leq \|\bv{A}\|\|\bv{x}^*\|$.
Therefore
\begin{align}
    \E\left[\|\bv{D} (\bv{u}_k-\bv{v}_k) \|^2 \right]
    &\leq 4\left(\frac{\|\bv{A}\|_\F^4}{RC} + \frac{\|\bv{A}\|_\F^2\|\bv{A}\|^2}{R}  +  \| \bv{A}\|^4\right) \|\bv{x}^* \|^2 
    \\&\notag\hspace{8em}+ 2\left( \frac{\|\bv{A}\|_\F^2\|\bv{A}\|^2}{R} \| + \|\bv{A}\|^2 \|\bv{A} \|^2 \right)\|\bv{x}^*\|^2
    \\&= \left( 4\frac{\|\bv{A}\|_\F^4}{RC}  + 6\frac{\|\bv{A}\|_\F^2\|\bv{A}\|^2}{R} +  6 \|\bv{A}\|^4  \right) \|\bv{x}^*\|^2.
\end{align}
Now, using that $\alpha^2 = 1/\|\bv{A}\|^2$, $R \geq 2\condF^2/\cond^2$, and $C \geq 10\condF^2 / \varepsilon^2$,
\begin{align}
    \alpha^2\E\left[\|\bv{D} (\bv{u}_k-\bv{v}_k)\|^2 \right]
    &\leq \left( 4\frac{\condF^4}{RC} + 6\frac{\condF^2\cond^2}{R} +  6 \cond^4  \right) \|\bv{x}^*\|^2
    \\&\leq \left( \frac{\varepsilon^2\cond^2}{5} + 3\cond^4 +  6\cond^4  \right) \|\bv{x}^*\|^2
    \\&\leq 10 \cond^4 \|\bv{x}^*\|^2,
\end{align}
where the last inequality uses that $\varepsilon < 1$ and $\cond \geq 1$.
\end{proof}

\subsubsection{Main bound}

\mainy*

\begin{proof}
By \cref{thm:y_cost}, 
\begin{equation}\label{eqn:phi_bd}
    \phi(\bv{y}_{k+1}) = \operatorname{nnz}(\bv{y}_{k+1})\frac{\sum_i \|\bv{a}_i \|^2 |(\bv{y}_{k+1})_i|^2}{\|\bv{A}^\T\bv{y}_{k+1}\|^2}
    \leq (kR) \frac{\|\bv{D}\bv{y}_{k+1}\|^2}{\|\bv{x}_{k+1}\|^2}.
\end{equation}
By the triangle inequality, and iterating the first bound, we obtain a bound
\begin{equation}
    \|\bv{D}\bv{y}_{k+1}\| \leq \|\bv{D} \bv{y}_k\| + 
    \alpha \|\bv{D} (\bv{u}_k-\bv{v}_k) \|
    \leq \alpha \sum_{j=0}^{k} \|\bv{D} (\bv{u}_j-\bv{v}_j)\|.
\end{equation}
Then, since $(c_1 + \cdots + c_p)^2 \leq (|c_1| + \cdots + |c_p|)^2 \leq p(c_1^2 + \cdots + c_p^2)$
\begin{equation}
    \|\bv{D}\bv{y}_{k+1}\|^2
    \leq \left(\alpha \sum_{j=0}^{k} \|\bv{D}(\bv{u}_j-\bv{v}_j)\|\right)^2
    \leq (k+1)\sum_{j=0}^{k} \alpha^2\|\bv{D}(\bv{u}_j-\bv{v}_j)\|^2.
\end{equation}
Now, using \cref{thm:Dg_bound},
\begin{equation}
    \E\left[\|\bv{D} \bv{y}_{k+1}\|^2 \right]
    \leq (k+1) \cdot 10 \cond^4 \|\bv{x}^*\|^2.
\end{equation}
By Markov's inequality, with probability at least $19/20$,
\begin{equation}\label{eqn:Dyk1_bound}
\|\bv{D} \bv{y}_{k+1}\|^2 \leq 
20\cdot 10 (k+1) \cond^4 \|\bv{x}^*\|^2
\leq \const\cdot k \cond^4 \|\bv{x}^*\|^2.
\end{equation}

Again by Markov's inequality, \cref{eqn:xk_bound} implies, with probability at least 19/20
\begin{equation}
\|\bv{x}_{k+1} - \bv{x}^*\|^2 \leq 20  \cdot 6 \varepsilon^2 \| \bv{x}^* \|^2 
= \const \cdot \varepsilon^2 \| \bv{x}^* \|^2
\leq \const \cdot \|\bv{x}^*\|^2.
\end{equation}
Then, by the triangle inequality,
\begin{equation}\label{eqn:xk1_xk_bound}
\|\bv{x}_{k+1}\|
\geq \|\bv{x}^*\| - \|\bv{x}_{k+1} - \bv{x}^*\| 
\geq \left(1 - \const\right)^{1/2} \|\bv{x}^*\|
\geq \const \cdot \ \|\bv{x}^*\|.
\end{equation}
Then, by a union bound, plugging \cref{eqn:xk1_xk_bound,eqn:Dyk1_bound} into \cref{eqn:phi_bd} gives the result.
\end{proof}

\begin{figure}[htb]
    \centering
    \includegraphics[scale=.8]{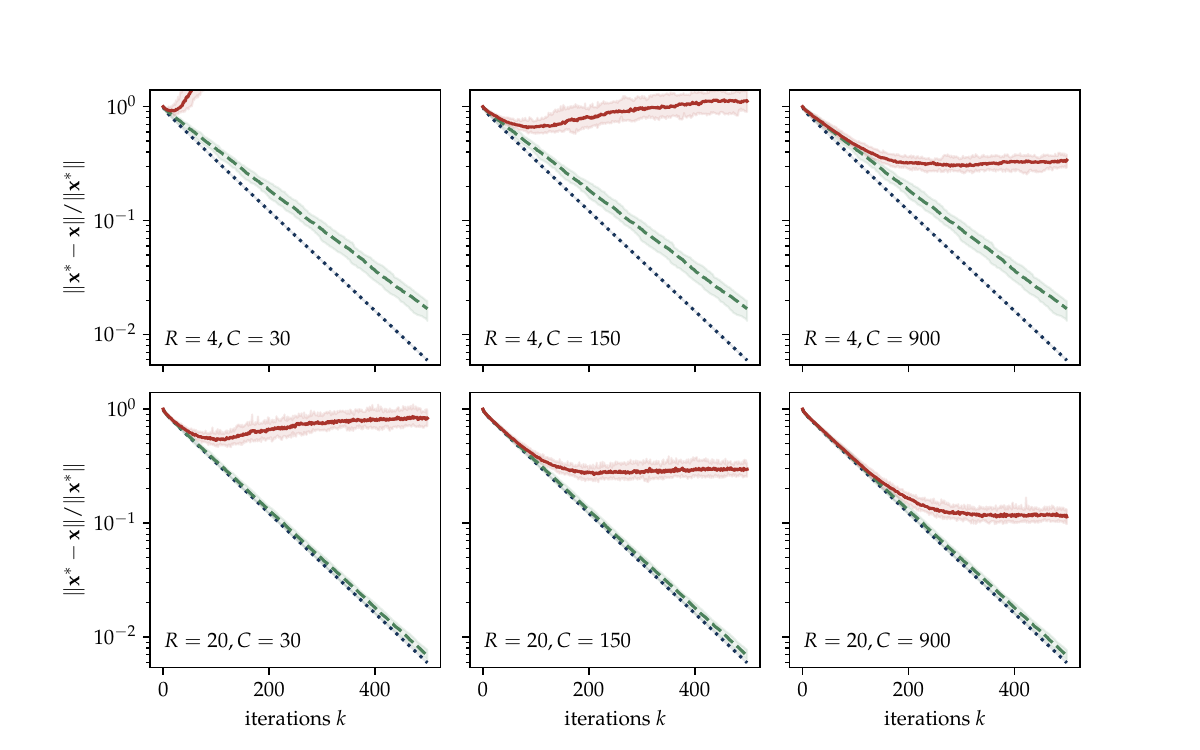}
    \caption{Median error and $5\%$-$95\%$ interval for iterates $\bv{x}_{k+1}$ \cref{eqn:xk_update} (\includegraphics[scale=.8]{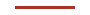}), $\E[\bv{x}_{k+1} | \mathcal{R}_0, \ldots, \mathcal{R}_k]$ (\includegraphics[scale=.8]{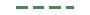}), and $\E[\bv{x}_{k+1}]$ (\includegraphics[scale=.8]{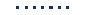}) as a function of iteration $k$.
    Note that the latter two iterates are equivalent to randomized Kaczmarz with averaging and gradient descent respectively.
    }
    \label{fig:convergence}
\end{figure}





\section{Numerical experiments}

\begin{figure}[tb]
    \centering
    \includegraphics[scale=.8]{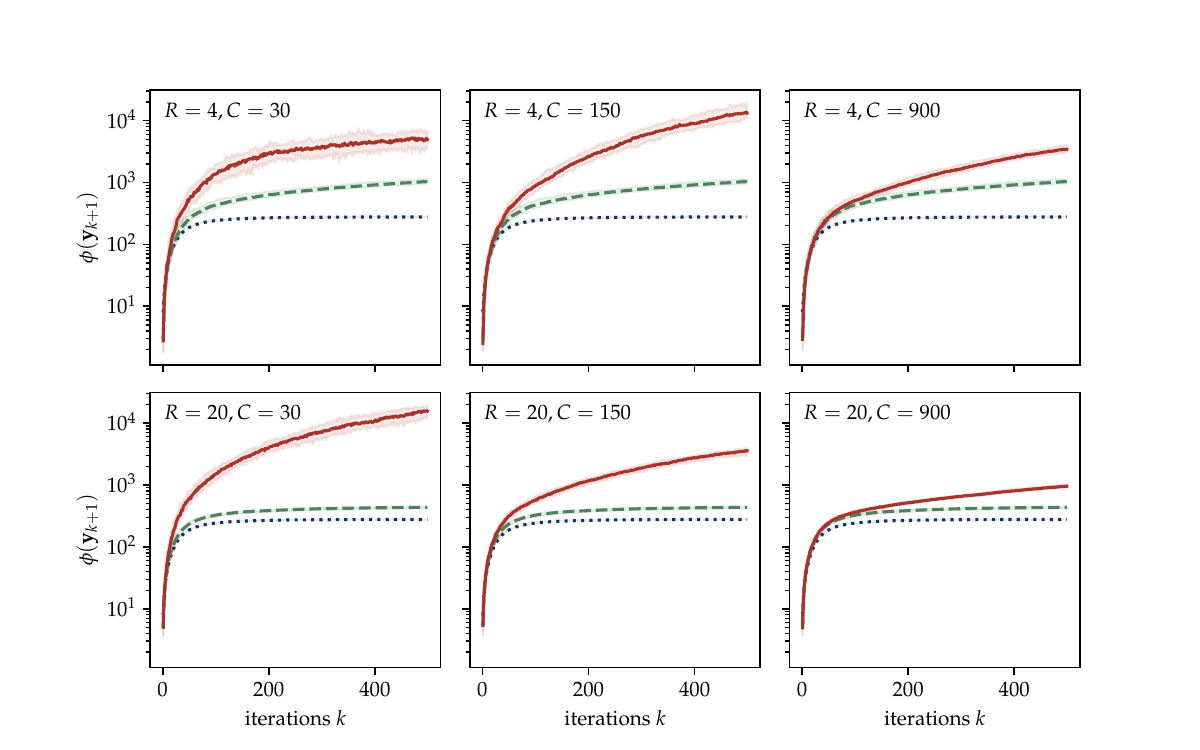}
    \caption{Median value and $5\%$-$95\%$ interval of sampling parameter $\phi(\bv{y}_{k+1})$ for iterates $\bv{y}_{k+1}$ \cref{eqn:xk_update} (\includegraphics[scale=.8]{imgs/solid.pdf}), $\E[\bv{y}_{k+1} | \mathcal{R}_0, \ldots, \mathcal{R}_k]$ (\includegraphics[scale=.8]{imgs/dash.pdf}), and $\E[\bv{y}_{k+1}]$ (\includegraphics[scale=.8]{imgs/dot.pdf}) as a function of iteration $k$.
    }
    \label{fig:sampling}
\end{figure}
We perform a basic numerical experiment on a test problem.
Here $\bv{A}$ is a $n\times d$ matrix with $n=3000$ and $d=2000$.
The left and right singular vectors are drawn from a Haar distribution. 
There are 100 nonzero singular values of the form $1+\rho_i$, where $\rho_i$ are geometrically spaced from $10^{-15}$ to $9$.
Hence, $\cond^2 \approx 100$ and $\condF^2 \approx 312.7$.
We run \cref{alg:main} to obtain $\bv{y}_{k+1}$.
We also compute the algorithm $\E[\bv{y}_{k+1} | \mathcal{R}_0, \ldots, \mathcal{R}_k]$ which corresponds to doing only row-sampling, and $\E[\bv{y}_{k+1}]$ which corresponds to no sampling.
These algorithms can be viewed as randomized Kaczmarz with averaging and gradient descent respectively.
The experiment is repeated 40 times to produce confidence intervals.

In \cref{fig:convergence} we look at the error $\|\bv{x}^* - \bv{x}_{k+1}\|/\|\bv{x}^*\|$, where $\bv{x}_{k+1} = \bv{A}^\T\bv{y}_{k+1}$ is computed exactly.
Increasing $R$ ensures the rate of convergence is similar to gradient descent, while increasing $C$ improves the maximal accuracy of the method.
The polynomial dependence of $C$ on the target accuracy $\varepsilon$ may be a limiting factor in practical situations.

Next, in \cref{fig:sampling} we plot the sampling parameter $\phi(\bv{y}_{k+1})$ which, as guaranteed by \cref{thm:sampling}, controls the cost of sampling from entries of $\bv{x}_{k+1}$.
Interestingly, we observe that this quantity seems to continue to increase, even when the iterate has reached its maximal accuracy.
This suggests that the analysis approach from \cref{sec:y_bound}, which results in a bound that does not converge to some limit as the number of iterations increases, may be reasonable.

\section{Outlook}

We have provided an elementary analysis of a randomized (quantum-inspired) algorithm for sampling from the solution to a linear system of equations.
Our analysis simplifies the analyses in a number of past works.
For simplicity, we opted only to consider consistent systems.
One direction for future work would be to extend our analysis to inconsistent systems and ridge-regression problems. 
It would also be interesting to understand whether the use of momentum or other acceleration techniques can improve the rates obtained in this paper.

\section*{Acknowledgments}
The authors would like to thank Atithi Acharya, Brandon Augustino, Pradeep Niroula, and Raymond Putra for feedback in early stages of this work.
The authors also thank Changpeng Shao and Ashley Montanaro, for helpful discussions on their work.

\vfill

\section*{Disclaimer}
This paper was prepared for informational purposes by the Global Technology Applied Research center of JPMorgan Chase \& Co. This paper is not a merchandisable/sellable product of the Research Department of JPMorgan Chase \& Co. or its affiliates. Neither JPMorgan Chase \& Co. nor any of its affiliates makes any explicit or implied representation or warranty and none of them accept any liability in connection with this paper, including, without limitation, with respect to the completeness, accuracy, or reliability of the information contained herein and the potential legal, compliance, tax, or accounting effects thereof. This document is not intended as investment research or investment advice, or as a recommendation, offer, or solicitation for the purchase or sale of any security, financial instrument, financial product or service, or to be used in any way for evaluating the merits of participating in any transaction.

\clearpage

\printbibliography

\appendix

\section{Deferred proofs}
\label{sec:deferred_proofs}

In this section, we provide the proofs of the key lemmas used above. 
These proofs are mostly just straightforward (albeit tedious) computations of expectations.

\subsection{Basic facts}

We will repeatedly use a basic fact about iid sums.
\begin{lemma}\label{thm:expected_norm_sum}
Let $\bv{Z}_1, \ldots \bv{Z}_m$ iid copies of a random matrix $\bv{Z}$.
Then
\begin{equation*}
    \E\left[\left(\frac{1}{m}\sum_{i=1}^{m} \bv{Z}_i \right)^\T\left(\frac{1}{m}\sum_{i=1}^{m} \bv{Z}_i \right) \right]
    = \frac{1}{m} \mathbb{E}\left[ \bv{Z}^\T\bv{Z} \right] + \left(1 - \frac{1}{m}\right) \E \left[\bv{Z}\right] ^\T\E \left[\bv{Z}\right]
\end{equation*}
and hence
\begin{equation*}
    \E\left[\left\|\frac{1}{m}\sum_{i=1}^{m} \bv{Z}_i \right\|_\F^2 \right]
    = \frac{1}{m} \mathbb{E}\left[ \|\bv{Z}\|_\F^2 \right] + \left(1 - \frac{1}{m}\right) \|\E \left[\bv{Z}\right]\|_\F^2.
\end{equation*}
\end{lemma}


We also recall a bound for approximate matrix multiplication; see e.g. \cite[Lemma 4]{drineas_kannan_mahoney_06}.
\begin{imptheorem}
\label{thm:approxMM}
Let $\mathcal{S}$ be an ordered list of $S$ indices each sampled with replacement so that, for each $s\in \mathcal{S}$, $\Pr[s = i] = p_i$, $i=1, \ldots, m$.
Given matrices $\bv{X}$ and $\bv{Y}$ with shared inner dimension $m$, define
\begin{equation*}
\bv{Z} = \sum_{s\in \mathcal{S}} \frac{1}{p_s}\bv{X}\bv{e}_s \bv{e}_s^\T\bv{Y}.
\end{equation*}
Then $\E[\bv{Z}] = \bv{X}\bv{Y}$ and
\begin{equation*}
\E \left[ \| \bv{X}\bv{Y} - \bv{Z} \|_\F^2 \right]
= \frac{1}{S} \sum_{i=1}^{m} \frac{\|\bv{X} \bv{e}_i\|^2\|\bv{e}_i^\T \bv{Y}\|^2}{p_i} 
- \frac{1}{S}\|\bv{X}\bv{Y}\|_\F^2.
\end{equation*}
\end{imptheorem}

\subsection{Proofs of accuracy bounds}
\label{sec:deferred_proofs:accuracy}

\Mkbd*
\begin{proof}
We drop the subscript $k$ for simplicity.
    First note that 
    \begin{equation}
        \bv{M} = \frac{1}{R} \sum_{r\in \mathcal{R}} \frac{1}{p_r}\bv{A}^\T \bv{e}_r \bv{e}_r^\T \bv{A}
    \end{equation}
    Clearly $\E[p_r^{-1} \bv{a}_{r} \bv{a}_{r}^\T] = \bv{A}^\T\bv{A}$.
    Now, observe that 
    \begin{align*}
        \E\left[ \left( \frac{1}{p_r} \bv{A}^\T\bv{e}_r\bv{e}_r^\T\bv{A} \right)^2 \right]
        = \E\left[ \frac{\|\bv{e}_r^\T\bv{A}\|^2}{p_{r}^2} \bv{a}_{r} \bv{a}_{r}^\T \right]
        = \sum_{r=1}^{n} \frac{\|\bv{e}_r^\T\bv{A}\|^2}{p_r} \bv{A}^\T\bv{e}_r\bv{e}_r^\T\bv{A}
        = \|\bv{A}\|_\F^2 \bv{A}^\T\bv{A}.
    \end{align*}
    The result then follows by \cref{thm:expected_norm_sum}.
\end{proof}

\colboundsingle*

\begin{proof}
We drop the subscripts on $k$ for simplicity.
First, note that 
\begin{equation}
    \bv{M} = \frac{1}{R} \sum_{r\in \mathcal{R}} \frac{1}{p_r}\bv{A}^\T \bv{e}_r \bv{e}_r^\T \bv{A}
    ,\quad
    \bv{z} = \frac{1}{C}\sum_{c\in\mathcal{C}}\frac{1}{q_c}\bv{M}\bv{e}_c\bv{e}_c^\T\bv{x}.
\end{equation}
Then, by \cref{thm:approxMM} (with $\bv{X} = \bv{M}$ and $\bv{Y} = \bv{x}$), 
\begin{equation}
    \E[ \| \bv{M}\bv{x} - \bv{z}\|^2 | \bv{M},\bv{x} ] 
    \leq \frac{1}{C} \sum_{c=1}^{d} \frac{\|\bv{M}\bv{e}_c\|^2\| \bv{e}_c^\T\bv{x}\|^2}{q_c} 
    = \frac{\|\bv{A}\|_\F^2}{C} \sum_{c=1}^{d} \frac{\|\bv{M}\bv{e}_c\|^2\| \bv{e}_c^\T\bv{x}\|^2}{\|\bv{A}\bv{e}_c\|^2}.
\end{equation}
Next, using the definition of $\bv{M}$ and again by \cref{thm:approxMM} (with $\bv{X} = \bv{A}^\T$ and $\bv{Y} = \bv{A}\bv{e}_c$),
\begin{equation}
    \E[ \| \bv{A}^\T \bv{A}\bv{e}_c - \bv{M}\bv{e}_c \|^2 ]
    \leq \frac{1}{R} \sum_{r=1}^{n} \frac{\|\bv{e}_r^\T\bv{A}\|_\F^2 |\bv{e}_r^\T\bv{A}\bv{e}_c|^2}{p_r}
    = \frac{\|\bv{A}\|_\F^2}{R} \sum_{r=1}^{n} |\bv{e}_r^\T\bv{A}\bv{e}_c|^2
    =\frac{\|\bv{A}\|_\F^2\|\bv{A}\bv{e}_c\|^2}{R}.
\end{equation}
Hence, since $\E[\bv{M}\bv{e}_c] = \bv{A}^\T\bv{A}\bv{e}_c$,
\begin{equation}
    \E[ \| \bv{M}\bv{e}_c \|^2 ]
    = \E[ \| \bv{M}\bv{e}_c - \bv{A}^\T\bv{A}\bv{e}_c \|^2 ] + \| \bv{A}^\T\bv{A}\bv{e}_c \|^2
    \leq
    \frac{\|\bv{A}\|_\F^2\|\bv{A}\bv{e}_c\|^2}{R} + \|\bv{A}\|^2\|\bv{A}\bv{e}_c\|^2
    .
\end{equation}
Combining the above equations, we get a bound 
\begin{equation}
    \E[ \| \bv{M}\bv{x} -  \bv{z}\|^2 | \bv{x} ] 
    \leq \frac{\|\bv{A}\|_\F^2}{C} \left( 
 \frac{\|\bv{A}\|_\F^2 \|\bv{x}\|^2}{R} + \|\bv{A}\|^2\|\bv{x}\|^2\right),
\end{equation}
as desired.
\end{proof}

\subsection{Proofs of sampling bounds}
\label{sec:deferred_proofs:sampling}

To prove \cref{thm:Dg_bound}, we use two lemmas, which we will prove in a moment.

\Dgtermabound*
\begin{proof}
We drop the subscripts on $k$ for simplicity.
Define 
\[
\hat{\bv{u}}_r := \frac{1}{p_r}\bv{D}\bv{e}_r \left(\frac{1}{C}\sum_{c\in\mathcal{C}_{k}} \frac{1}{q_{c}} (\bv{e}_r^\T\bv{A}\bv{e}_{c})( \bv{e}_c^\T \bv{x}) \right).
\]
First, observe that
\begin{equation}
    \E\left[ \hat{\bv{u}}_r \right]
    =\E\left[ \E\left[\frac{1}{p_r} \bv{D}\bv{e}_r \left(\frac{1}{C}\sum_{c\in\mathcal{C}_{k}} \frac{1}{q_{c}} (\bv{e}_r^\T\bv{A}\bv{e}_{c})( \bv{e}_c^\T \bv{x}) \right) \middle| r \right]\right]
    = \E\left[ \frac{1}{p_r}\bv{D}\bv{e}_r \bv{e}_r^\T\bv{A}\bv{x} \right]
    = \bv{D} \bv{A} \bv{x}
\end{equation}
and hence, since $\|\bv{D}\| \leq \|\bv{A}\|$,
\begin{equation}
    \| \E[ \hat{\bv{u}}_r  | \bv{x} ]\|^2
    = \| \bv{D} \bv{A} \bv{x} \|^2
    \leq \| \bv{A}\|^4 \|\bv{x} \|^2.
\end{equation}
Now, note that by \cref{thm:approxMM} (with $\bv{X} = \bv{e}_r^\T \bv{A}$ and $\bv{Y} = \bv{x}$),
\begin{align}
    \E\left[\left|\frac{1}{C}\sum_{c\in\mathcal{C} }\frac{1}{q_{c}} \bv{e}_r^\T \bv{A} \bv{e}_{c} \bv{e}_c^\T \bv{x} \right|^2 \middle| r,\bv{x}  \right]
    &\leq \frac{1}{C}\sum_{c=1}^{d}\frac{|\bv{e}_r^\T \bv{A} \bv{e}_{c}|^2|\bv{e}_c^\T \bv{x}|^2}{q_c} + \|\bv{e}_r^\T \bv{A} \bv{x}\|^2
    \\&= \frac{\|\bv{A}\|_\F^2}{C}\sum_{c=1}^{d}\frac{|\bv{e}_r^\T \bv{A} \bv{e}_{c}|^2|\bv{e}_c^\T \bv{x}|^2}{\|\bv{A}\bv{e}_c\|^2} + \| \bv{e}_r^\T\bv{A} \bv{x} \|^2.
\end{align}
Then, since $\bv{D}\bv{e}_r = \|\bv{e}_r^\T\bv{A}\|^2 \bv{e}_r$,
\begin{align}
    \E\left[ \|\hat{\bv{u}}_r\|^2  \middle| \bv{x} \right]
    &= \E\left[\frac{1}{p_r^2} \left\|\bv{D}\bv{e}_r \left( \frac{1}{C}\sum_{c\in\mathcal{C}_{j}} \frac{1}{q_{c}} (\bv{e}_r^\T \bv{A} \bv{e}_{c})( \bv{e}_c^\T\bv{x}) \right)^2 \right\|^2  \middle| \bv{x} \right]
    \\&= \E\left[\frac{\|\bv{e}_r^\T\bv{A}\|^2}{p_r^2} \E\left[\left|\frac{1}{C}\sum_{c\in\mathcal{C}} \frac{1}{q_{c}} \bv{e}_r^\T \bv{A} \bv{e}_{c} \bv{e}_c^\T\bv{x} \right|^2 \middle| r,\bv{x} \right] \middle| \bv{x} \right]
    \\&\leq \E\left[\frac{\|\bv{e}_r^\T\bv{A}\|^2}{p_r^2} \left( \frac{\|\bv{A}\|_\F^2}{C}\sum_{c=1}^{d}\frac{|\bv{e}_r^\T \bv{A} \bv{e}_{c}|^2|\bv{e}_c^\T\bv{x}|^2}{\|\bv{A}\bv{e}_c\|^2} + \|\bv{e}_r^\T\bv{A} \bv{x} \|^2 \right)\middle| \bv{x} \right]
    \\&= \sum_{r=1}^{n} \frac{\|\bv{e}_r^\T\bv{A}\|^2}{p_r} \left(  \frac{\|\bv{A}\|_\F^2}{C}\sum_{c=1}^{d}\frac{|\bv{e}_r^\T \bv{A} \bv{e}_{c}|^2|\bv{e}_c^\T\bv{x}|^2}{\|\bv{A}\bv{e}_c\|^2} + \|\bv{e}_r^\T\bv{A} \bv{x} \|^2 \right)
    \\&= \underbrace{\frac{ \|\bv{A}\|_\F^4}{C} \sum_{r=1}^{n}  \sum_{c=1}^{d}\frac{|\bv{e}_r^\T \bv{A} \bv{e}_{c}|^2|(\bv{x}_j)_c|^2}{\|\bv{A}\bv{e}_c\|^2}}_{\text{term a}} + \underbrace{\|\bv{A}\|_\F^2\sum_{r=1}^{n} \|\bv{e}_r^\T\bv{A} \bv{x} \|^2}_{\text{term b}} .
    \label{eqn:term1_two_terms}
\end{align}
Next we bound ``term a'' in \cref{eqn:term1_two_terms}.
Since $\|\bv{A}\bv{e}_c\|^2 = \sum_{r=1}^{n} |\bv{e}_r^\T \bv{A} \bv{e}_{c}|^2$,
\begin{align}
    \frac{\|\bv{A}\|_\F^4}{C} \sum_{r=1}^{n}  \sum_{c=1}^{d}\frac{|\bv{e}_r^\T \bv{A} \bv{e}_{c}|^2|\bv{e}_c^\T\bv{x}|^2}{\|\bv{A}\bv{e}_c\|^2}
    = \frac{\|\bv{A}\|_\F^4}{C}\sum_{c=1}^{d}\frac{\|\bv{A}\bv{e}_c\|^2|\bv{e}_c^\T\bv{x}|^2}{\|\bv{A}\bv{e}_c\|^2} 
    = \frac{\|\bv{A}\|_\F^4}{C} \|\bv{x}\|^2
\end{align}
Now, we bound ``term b'' in \cref{eqn:term1_two_terms}:
\begin{align}
    \|\bv{A}\|_\F^2\sum_{r=1}^{n} \|\bv{e}_r^\T\bv{A} \bv{x} \|^2
    = \|\bv{A}\|_\F^2 \|\bv{A}\bv{x}\|^2
    \leq \|\bv{A}\|_\F^2 \|\bv{A}\|^2 \|\bv{x}\|^2.
\end{align}
Hence, by \cref{thm:expected_norm_sum},
\begin{align}
    \E\left[ \left\| \bv{D}\bv{u} \right\|^2 \middle| \bv{x} \right]
    &= \frac{1}{R} \E\left[ \|\hat{\bv{u}}_r \|^2 \middle| \bv{x} \right]
    + \left(1 - \frac{1}{R} \right) \left\| \E\left[ \hat{\bv{u}}_r \middle| \bv{x} \right] \right\|^2
    \\&\leq \frac{\|\bv{A}\|_\F^4}{RC} \|\bv{x}\|^2 + \frac{\|\bv{A}\|_\F^2\|\bv{A}\|^2}{R}\|\bv{x}\|^2 +  \| \bv{A}\|^4 \|\bv{x} \|^2,\label{eqn:Dg_term1_bound}
\end{align}
as desired.
\end{proof}

\Dgtermbbound*
\begin{proof}
Define 
\[
\hat{\bv{v}}_r := \frac{1}{p_r} \bv{D}\bv{e}_r \bv{b}_r.
\]
\begin{equation}
    \E[ \hat{\bv{v}}_r ]
    = \E\left[ \frac{1}{p_r}\bv{D}\bv{e}_r \bv{b}_r\right]
    = \sum_{r=1}^{n} \bv{D}\bv{e}_r \bv{b}_r
    = \bv{D} \bv{b}
\end{equation}
and hence, since $\|\bv{D}\| \leq \|\bv{A}\|$,
\begin{equation}
    \| \E[ \hat{\bv{v}}_r  | \bv{x} ]\|^2
    = \| \bv{D} \bv{b} \|^2
    \leq \| \bv{A}\|^2 \|\bv{b} \|^2.
\end{equation}
Likewise
\begin{align}
    \E\left[ \| \hat{\bv{v}}_r \|^2  \right]
    &= \E\left[ \frac{1}{p_r^2} \| \bv{D} \bv{e}_r \bv{b}_r\|^2   \right]
    \\&= \E\left[ \frac{\|\bv{e}_r^\T\bv{A}\|^2}{p_r^2}  |\bv{b}_r|^2   \right]
    \\&= \sum_{r=1}^{n} \frac{\|\bv{e}_r^\T\bv{A}\|^2}{p_r}  |\bv{b}_r|^2  
    \\&= \|\bv{A}\|_\F^2 \sum_{r=1}^{n}|\bv{b}_r|^2  
    \\&=  \|\bv{A}\|_\F^2 \|\bv{b}\|^2.
\end{align}
Hence, by \cref{thm:expected_norm_sum},
\begin{equation}
    \E\left[ \left\| \bv{D} \bv{v}_k \right\|^2 \middle| \bv{x}_k \right]
    = \frac{1}{R} \E\left[ \| \hat{\bv{v}}_r \|^2  \right]
    + \left(1 - \frac{1}{R}\right) \left\| \E\left[ \hat{\bv{v}}_r \right] \right\|^2
    \leq \frac{\|\bv{A}\|_\F^2}{R} \|\bv{b}\|^2 +  \| \bv{A}\|^2 \|\bv{b} \|^2
    \label{eqn:Dg_term2_bound}
\end{equation}
as desired.
\end{proof}





\section{Sampling}
\label{sec:vector_sampling}

For completeness, we describe the algorithm underlying \cref{thm:sampling} which allows us to sample from $\bv{x} = \bv{A}^\T\bv{y}$, given a sparse vector $\bv{y}$ and appropriate sample/query access to $\bv{A}$.
Our exposition follows \cite[\S4]{tang_19}.

\subsection{Rejection sampling}

Rejection sampling is a standard technique for converting samples from a distribution $P$ into samples from a distribution $Q$, given access to a function $r(j) = Q(j) / (MP(j))$, where $M$ is some constant for which $Q(j) \leq M P(j)$.
The algorithm is described in \cref{alg:rejection} and requires an expected $M$ iterations to terminate.

\begin{algorithm}[ht!]
\caption{Rejection sampling}
\label{alg:rejection}
\begin{algorithmic}[1]
\Require Sample access to $P$, function $r(j)$
    \While{true}
    \State Sample $j\sim P$
    \State break with probability $r(j)$
    \EndWhile
\Ensure $j \sim Q$
\end{algorithmic}
\end{algorithm}

\subsection{Sampling from a linear combination}

We would like to sample from $Q = \mathcal{D}_{\bv{x}}$.
We can write $Q$ explcitly as 
\begin{equation}
Q(s) = \frac{|\bv{a}_{*,s}^\T\bv{y}|^2}{\|\bv{A}^\T\bv{y}\|^2}.
\end{equation}
While we can evaluate $Q(j)$ for a single value of $j$ efficiently, sampling from $Q$ is harder.
To do this, we will make use of rejection sampling with 
\begin{equation}
P(j) := 
\frac{\sum_{i} |y_i|^2|\bv{a}_{i,j}|^2}{\sum_{i}|y_{i}|^2\|\bv{a}_{i}\|^2} .
\end{equation}
Then, setting
\begin{equation}
M 
:= \frac{s\sum_{i} |y_i| \| \bv{a}_i\|^2}{\|\bv{A}^\T\bv{y}\|^2}
= \phi(\bv{y}),
\end{equation}
we see that we can compute
\begin{equation}
r(j) 
:= \frac{Q(j)}{M P(j)}
= \frac{|\bv{a}_{*,j}^\T\bv{y}|^2}{\|\bv{A}^\T\bv{y}\|^2}
\cdot \frac{\|\bv{A}^\T\bv{y}\|^2}{j\sum_{i} |y_i| \| \bv{a}_i\|^2} \cdot \frac{\sum_{i}|y_{i}|^2\|\bv{a}_{i}\|^2}{\sum_{i} |y_i|^2|\bv{a}_{i,j}|^2}
= \frac{|\bv{a}_{*,j}^\T\bv{y}|^2}{s\sum_{i} |y_i|^2|\bv{a}_{i,j}|^2},
\end{equation}
using $O(s)$ time and queries to $\bv{A}$.

\begin{algorithm}[ht!]
\caption{Sampling from a linear combination}
\label{alg:P_def}
\begin{algorithmic}[1]
\Require $\bv{y}$, sample/query access to $\bv{A}$
\State Sample row $r$ so that 
$\Pr[r = i] \propto |y_i|\|\bv{a}_i\|^2$
\Comment{Use rejection sampling}
\State Sample a column $c \sim \mathcal{D}_i^{\textup{col}}$ \Comment{$\Pr[c = j | r=i] \propto |\bv{a}_{i,j}|^2$}
\Ensure $c\sim P$
\end{algorithmic}
\end{algorithm}

To generate samples from $P$, we use \cref{alg:P_def}.
The correctness of this algorithm is verified by a direct computation. 
Indeed, by the law of total probability,
\begin{equation}
P(j)
= \Pr[c = j] 
= \sum_{i}  \Pr[a = j | r = i] \Pr[r = i]
= \sum_{i} \frac{|\bv{a}_{i,j}|^2}{\|\bv{a}_i\|^2}\cdot \frac{|y_i|^2\|\bv{a}_i\|^2}{\sum_{i'}|y_{i'}|^2\|\bv{a}_{i'}\|^2}
= \frac{\sum_{i} |y_i|^2|\bv{a}_{i,j}|^2}{\sum_{i}|y_{i}|^2\|\bv{a}_{i}\|^2} .
\end{equation}

First, we must each $|y_i|\|\bv{a}_i\|^2$, which requires $O(s)$ time (assuming access to the relevant row norms).
The second step requires a single sample from $\bv{A}$.
The candidate sample $c$ is then used by rejection sampling.
Computing $r(c)$ requires $O(s)$ queries/time.
Finally, rejection sampling runs for an expected $O(\phi(\bv{y}))$ steps, resulting in the $O(s\phi(\bv{y}))$ runtime.

\section{Possible issues with \cite{shao_montanaro_22,shao_montanaro_23}}
\label{sec:SM_error}

The rates claimed in \cite{shao_montanaro_22,shao_montanaro_23} are similar to or better than the rates we prove in this paper. 
However, we were unable to follow several key steps in the paper. In the remainder of this section, we outline some apparent gaps in the analysis.

In a personal communication~\cite{shao_montanaro_25_personal} with the authors, they confirmed the following issues. In particular, they acknowledged that the techniques of~\cite{shao_montanaro_23} do not directly lead to the claimed rate of $\widetilde{O}(\condF^4\cond^2 / \varepsilon^2)$. As a consequence, our analysis actually yields a better complexity than the aforementioned works.

In what follows, we follow the notation of \cite{shao_montanaro_22,shao_montanaro_23}; i.e. a $\dagger$ indicates the transpose and $\tilde{b}_i = b_i / \|A_{i*}\|$ and $\tilde{A}_{i*} = A_{i*} / \|A_{i*} \|$.

\subsection{Runtime}
First, it is somewhat unclear exactly what the algorithm being analyzed is. 
In particular, while equation (19) in \cite{shao_montanaro_23} provides an update for $\bv{x}_{k+1}$, it is in terms of some $\mathcal{T}_k$ and diagonal matrix ${D}_i$, neither of which are clearly defined.
As far as we can tell $\mathcal{T}_k$ is the same as mentioned near (9) for Kaczmarz with averaging; ``each index $i\in[m]$ is put into $\mathcal{T}_k$ with probability proportional to $\|A_{i*}\|^2$''.
Below (19), it is said that ``the definition of $D_i$ is similar to (12)''. In (12), a matrix $D$ is defined, which sub-samples $d$ columns proportional to the squared column norms.

In the proof of Theorem 13, the update for $\bv{y}_{k+1}$ is then given as 
\begin{equation}
    \bv{y}_{k+1} = \bv{y}_k + \frac{1}{2} \sum_{i\in\mathcal{T}_k} \frac{\tilde{b}_i - \langle \tilde{A}_{i*}|{D}_i {A}^\dagger|\bv{y}_k\rangle}{\|A_{i*}\|} \bv{e}_i.
\end{equation}
The size of $\mathcal{T}_k$ is $\widetilde{O}(\kappa_\F^2 / \kappa^2)$.
Hence, $\bv{y}_{k+1}$ is $s = \widetilde{O}(k \kappa_\F^2 /\kappa^2)$-sparse.
This agrees with the sparsity stated in the proof.

Next, note that the number of nonzero entries in $D_i$ is $d = \widetilde{O}(\kappa_\F^2 / \varepsilon^2)$.
For each $i$, the dominant cost is computing $\langle \tilde{A}_{i*}|{D}_i {A}^\dagger|\bv{y}_k\rangle$.
This requires $\widetilde{O}(ds) = \widetilde{O}(\kappa_\F^4 / (\kappa^2 \varepsilon^2))$ cost.
This must then be repeated for up to $|\mathcal{T}_k| = \widetilde{O}(\kappa_\F^2 / \kappa^2)$ values of $i$.
Thus, it seems that the cost of iteration $k$ is $\widetilde{O}(k \kappa_\F^6 / (\kappa^4 \varepsilon^2))$, and the total cost $\widetilde{O}(k^2 \kappa_\F^6 / (\kappa^4 \varepsilon^2)) = \widetilde{O}(\kappa_\F^6 / \varepsilon^2)$.
This is worse than the claimed $\widetilde{O}(\kappa_\F^4\kappa^2 / \varepsilon^2)$ cost.

\subsection{Sampling time}

In Appendix A, the authors aim to bound $\phi(\bv{y}_{k+1})$.
For any two vectors $\bv{a},\bv{b}$, the authors define $\langle \bv{a} | \bv{b} \rangle_{\Lambda} := \langle \bv{a} | \Lambda |\bv{b} \rangle$, where $\Lambda := \operatorname{diag}(\|\bv{a}_{1}\|^2,\ldots, \|\bv{a}_{n}\|^2)$.
The authors then perform a variance analysis, which, at some point, results in the following sequence:
\begin{align}\label{eqn:SM_error}
    \frac{\langle \bv{b} | \bv{y}_k \rangle_{\Lambda} - \|\bv{x}_k\|^2_{\Lambda}}{\|A\|_\F^2}
    = \frac{\langle \bv{x}^* | A^\dagger | \bv{y}_k \rangle_{\Lambda} - \|\bv{x}_k\|^2_{\Lambda}}{\|A\|_\F^2}   
    = \frac{\langle \bv{x}^* | \bv{x}_k \rangle_{\Lambda} - \|\bv{x}_k\|^2_{\Lambda}}{\|A\|_\F^2}   
    \leq \langle \bv{x}^* | \bv{x}_k \rangle - \|\bv{x}_k\|^2,
\end{align}
where in the above, the authors ``used the fact that for any two vectors $\bv{a},\bv{b}$ we have $|\langle \bv{a},\bv{b}\rangle_{\Lambda}| \leq \|\Lambda\|^2 | \langle \bv{a} | \bv{b}\rangle|$''.
We are unable to follow the chain of logic in \cref{eqn:SM_error}.
In particular:
\begin{itemize}[itemsep=.1em,topsep=-.5em]
    \item 
    The quantity $\langle \bv{x}^* | \bv{x}_k \rangle_{\Lambda}$ is not well-defined; both $\bv{x}^*$ and $\bv{x}_k$ are length $d$ vectors, while $\Lambda$ is a $n\times n $ matrix.
    The mistake seems to be assuming that $A^\dagger$ and $\Lambda$ commute (which they need not, even if $n=d$).
    \item 
    The ``fact'' above is not true (even disregarding the inconsistent scaling on $\Lambda$).
    For instance, consider $\bv{a} = (1,-1)$, $\bv{b} = (1,1)$, and $\Lambda = \operatorname{diag}(2,1)$.
    Then $|\langle\bv{a}|\bv{b}\rangle_{\Lambda}| = 1$ while $\langle\bv{a}|\bv{b}\rangle = 0$.
    In fact, this issue feels very reminiscent of the fact that $\|\bv{A}^\dagger \bv{y}_{k+1} \|= \|\sum_i (\bv{y}_{k+1})_r \bv{a}_{r} \|$ can be smaller than $\sum_r |(\bv{y}_{k+1})_r|^2 \|\bv{e}_r^\T\bv{A}\|^2$.
\end{itemize}

\end{document}